\documentclass[11pt]{article}

\usepackage[utf8]{inputenc}
\usepackage[T1]{fontenc}
\usepackage[english]{babel}

\usepackage{thmtools}
\usepackage{thm-restate}

\usepackage{amsmath,amssymb,amsfonts,amsthm}
\usepackage{mathabx}

\usepackage{cleveref}

\newtheorem{lemma}{Lemma}
\newtheorem{theorem}{Theorem}

\theoremstyle{definition}
\newtheorem{definition}{Definition}

\crefname{lemma}{Lemma}{Lemmas}
\Crefname{lemma}{Lemma}{Lemmas}
\crefname{theorem}{Theorem}{Theorems}
\Crefname{theorem}{Theorem}{Theorems}
\crefname{corollary}{Corollary}{Corollaries}
\Crefname{corollary}{Corollary}{Corollaries}
\crefname{observation}{Observation}{Observations}
\Crefname{observation}{Observation}{Observations}
\crefname{definition}{Definition}{Definitions}
\Crefname{definition}{Definition}{Definitions}
\crefname{section}{Section}{Sections}
\Crefname{section}{Section}{Sections}
\crefname{figure}{Figure}{Figures}
\Crefname{figure}{Figure}{Figures}

\renewcommand{\subset}{\subseteq}

\newcommand{\abs}[1]{\left | #1 \right |}
\newcommand{\set}[1]{\left \{ #1 \right \}}

\newcommand{\Oh}[1]{O\!\left(#1\right)}
\newcommand{\Omegah}[1]{\Omega\!\left(#1\right)}

\newcommand{\Otildeh}[1]{\tilde{O}\!\left(#1\right)}
\newcommand{\Thetah}[1]{\Theta\!\left(#1\right)}
\newcommand{\poly}[0]{\operatorname{poly}}

\usepackage{listings}
\usepackage{color}
\usepackage[pdftex]{graphicx}
\usepackage[usenames,dvipsnames,table]{xcolor}
\definecolor{shade}{RGB}{235,235,235}

\usepackage{booktabs}

\usepackage{todonotes}

\usepackage{authblk}

\usepackage{algorithm2e}

\title{Additive Spanners and Distance Oracles in Quadratic Time}
\author{Mathias Bæk Tejs Knudsen\thanks{
	Research partly supported by
	Advanced Grant DFF-0602-02499B from the Danish Council for Independent Research
	under the Sapere Aude research career programme
	and by the FNU project AlgoDisc - Discrete Mathematics, Algorithms, and Data Structures}}
\affil{University of Copenhagen,\\
    \tt{mathias@tejs.dk}
}
\date{}

\begin{document}
\maketitle
\begin{abstract}
Let $G$ be an unweighted, undirected graph.
An additive $k$-spanner of $G$ is a subgraph $H$ that approximates
all distances between pairs of nodes up to an additive error of $+k$, that is, it satisfies
$d_H(u,v) \le d_G(u,v)+k$ for all nodes $u,v$, where $d$ is the shortest path distance.
We give a deterministic algorithm that constructs an additive $\Oh{1}$-spanner with $\Oh{n^{4/3}}$ edges in
$\Oh{n^2}$ time. This should be compared with the randomized Monte Carlo algorithm by Woodruff [ICALP 2010]
giving an additive $6$-spanner with $\Oh{n^{4/3}\log^3 n}$ edges in expected time $\Oh{n^2\log^2 n}$.

An $(\alpha,\beta)$-approximate distance oracle for $G$ is a data structure that supports the
following distance queries between pairs of nodes in $G$. Given two nodes
$u$, $v$ it can in constant time compute a distance estimate $\tilde{d}$ that satisfies
$d \le \tilde{d} \le \alpha d + \beta$ where $d$ is the distance between $u$ and $v$ in
$G$.
Sommer [ICALP 2016] gave a randomized Monte Carlo $(2,1)$-distance oracle of size $\Oh{n^{5/3}\poly\log n}$
in expected time $\Oh{n^2\poly\log n}$. As an application of the additive $\Oh{1}$-spanner we improve the construction
by Sommer [ICALP 2016] and
give a Las Vegas $(2,1)$-distance oracle of size $\Oh{n^{5/3}}$ in time $\Oh{n^2}$. This also implies an
algorithm that in $\Oh{n^2}$ gives approximate distance for all pairs of
nodes in $G$ improving on the $\Oh{n^2 \log n}$ algorithm by Baswana and Kavitha [SICOMP 2010].
\end{abstract}

\thispagestyle{empty}

\newpage
\setcounter{page}{1}
\section{Introduction}

%TODO:
%-Skrive ny start på intro
%-Skrive at mit ikke er mere kompliceret

Let $G = (V,E)$ be an unweighted, undirected graph on $n$ nodes and $m$ edges. A subgraph $H$ of
$G$ is an \emph{additive $k$-spanner} if the following holds for every pair $u,v$ of nodes in
$G$:
\begin{align}
	\notag
	d_H(u,v) \le d_G(u,v) + k
	\, ,
\end{align}
where $d_H(u,v)$ and $d_G(u,v)$ is the distance between $u$ and $v$ in $H$ and $G$ respectively.
This paper will only consider additive spanners and not multiplicative or mixed spanners, so we
will simply say that $H$ is a \emph{$k$-spanner} when we mean that $H$ is an additive $k$-spanner.

In this paper we consider algorithms constructing $k$-spanners, and there are therefore three
interesting parameters: The distortion $k$, the running time of the algorithm, and the size of the
spanner created.
Elkin and Peleg \cite{elkin01spanners} showed
how to construct $2$-spanners with $\Oh{n^{3/2}}$ edges in $\Oh{n^{5/2}}$ time, and
Baswana et al \cite{baswana10spanners} gave an algorithm that constructs $6$-spanners with $\Oh{n^{4/3}}$
edges in $\Oh{n^{2/3}m}$ time.

The running time of these algorithms can be improved if we allow the $k$-spanners to be larger by
a $\poly \log n$ factor. Dor, Halperin and Zwick \cite{dor00apsp} showed
that we can construct $2$-spanners with $O\!\left(n^{3/2}\log^{1/2} n\right)$ edges in 
$O\!\left(n^2\log^{2}n\right)$ time, and Woodruff \cite{woodruff10nearlyquadratic} gave an algorithm to construct $6$-spanners
with $O\!\left(n^{4/3}\log^3 n\right)$ edges in $O\!\left(n^2\log^2 n\right)$ time. The construction of
Woodruff is furthermore randomized Monte Carlo. These results are summarized in
Table \ref{tab:summaryTime}.

These improvements to the running time fit into the following paradigm: For a fixed $k$ the authors
find algorithms that produce spanners that are \emph{almost} as small as the best known construction of
$k$-spanners and have \emph{near}-quadratic running time. We reverse this way of looking at the
problem. We are now trying to find algorithms that yield $k$-spanners
that are exactly as small as the best known constructions for any $k=O(1)$, i.e. $O(n^{4/3})$, and at the same
time we want the algorithm to run as fast as possible. All known algorithms for creating $O(1)$-spanners that
have close to optimal size run in time $\Omega(n^2)$.
\footnote{For instance the algorithm by Baswana et al \cite{baswana10spanners} gives a $6$-spanner with
$\Oh{n^{4/3}}$ edges and is therefore only interesting when $m = \Omegah{n^{4/3}}$, in which case the
running time is $\Thetah{n^{2/3}m} = \Omegah{n^2}$.}
So a natural question is to ask if there exists a $k=O(1)$ and an algorithm that constructs a $k$-spanner
with $\Oh{n^{4/3}}$ edges in $\Oh{n^2}$ time. In fact Sommer \cite{sommer2016distanceOracle} mentioned at his talk at
ICALP 2016 that the main obstacle towards getting a better running time for constructing the distance
oracle he presented is the lack of such an algorithm. In his case the distortion $k = O(1)$ is only factored
into the running time and not the distortion of oracle. Therefore, it does not matter what $k$ is as long as it
is constant. 

We show that it possible to attain this goal by giving an algorithm that
constructs $8$-spanners deterministically with $\Oh{n^{4/3}}$ edges in $\Oh{n^2}$ time.
Comparing this with the algorithm by Woodruff \cite{woodruff10nearlyquadratic}
this gets rid of the $\log^3 n$ factor
%\footnote{
%The factor of $\log^3 n$ on the number of edges may not seem significant, but e.g. if $\log$ is
%the base $2$ logarithm then $n^{4/3}\log^3 n$ is larger than $n^{3/2}$ (the guarantee of the $2$-spanner)
%when $n \le 2^{125}$, i.e. for all realistic values of $n$. Obviously, this point of view is too simplified
%as it does not consider the constants in the $O$-notation, and more importantly only considers
%constant values of $n$. The point, however, is that a factor of $\log^3 n$ is not negligible.}
on the number of edges and a factor of $\log^2 n$ in the running time. Furthermore, the algorithm
is deterministic and not randomized Monte Carlo.
The price of these improvements is that the distortion is larger than $6$. We note that there
are no lower bounds ruling out the possibility of a $4$-spanner with $\Oh{n^{4/3}}$ edges.
For the application to the distance oracle by Sommer \cite{sommer2016distanceOracle},
the distortion is unimportant as long as it is constant. We also
show how to construct $2$-spanners with $\Oh{n^{3/2}}$ edges in $\Oh{n^2}$ time.
For a comparison to previous work see Table \ref{tab:summaryTime}.

\begin{table}[!h]
	\begin{center}
		\begin{tabular}{l|l|l|l|l}
			$k$ & \text{Number of Edges} & \text{Running Time} & \text{Comment} & \text{Reference} \\ \hline
			$2$ & $\Oh{n^{3/2}}$ & $\Oh{n^{5/2}}$ & Deterministic & \cite{elkin01spanners} \\
			$2$ & $\Oh{n^{3/2}\log^{1/2} n}$ & $\Oh{n^2\log^{2} n}$ & Deterministic & \cite{dor00apsp} \\
			$2$ & $\Oh{n^{3/2}}$ & $\Oh{n^2}$ & Deterministic & Theorem \ref{thm:twospanner} \\
			\hline
			$6$ & $\Oh{n^{4/3}}$ & $\Oh{n^{2/3}m}$ & Deterministic & \cite{baswana10spanners} \\
			$6$ & $\Oh{n^{4/3}\log^{3} n}$ & $\Oh{n^2\log^{2} n}$ & Randomized Monte Carlo & \cite{woodruff10nearlyquadratic} \\
			$8$ & $\Oh{n^{4/3}}$ & $\Oh{n^2}$ & Deterministic & Theorem \ref{thm:eightspanner} \\
		\end{tabular}
	\end{center}
	\caption{A summary of the performance of selected algorithms that creates a $k$-spanner $H$ from a graph
	on $n$ nodes. It shows the additive distortion, $k$, and an upper bound on the number of edges in $H$ as
	well as the running time of the algorithm that constructs $H$.}
	\label{tab:summaryTime}
\end{table}

\subparagraph{Related work}

Elkin and Peleg \cite{elkin01spanners} showed that\footnote{
Aingworth et al \cite{aingworth99fastestimationshortestpath} earlier showed the same
result up to logarithmic factors on the size of the spanner.} any graph on $n$ nodes has a $2$-spanner with $O(n^{3/2})$ edges,
Chechik \cite{chechik13spanners} showed that it has a $4$-spanner with $O\!\left(n^{7/5}\log^{1/5} n\right)$ edges, and
Baswana et al \cite{baswana10spanners} showed that it has a $6$-spanner with $O(n^{4/3})$ edges.
These results are complemented by a negative result of Abboud and Bodwin \cite{abboud16tightexponent}. A consequence of their
result is that for any $k = O(1)$ there exists a graph on $n$ nodes such that any $k$-spanner of this
graph has at least $n^{4/3-o(1)}$ edges.
%The $n^{o(1)}$ term is $2^{\Theta(\sqrt{\log n})}$.

Another negative result comes from Erd\H{o}s's girth conjecture \cite{erdos1964extremal}. It states that for any constant $k$
there exists graphs with $n$ nodes and $\Omega\!\left(n^{1+1/k}\right)$ edges where the girth is
$2k+2$. This conjecture has been proved for $k=2,3,5$ \cite{wenger1991extremal,benson1966minimal}. In particular
if the conjecture is true this implies that there exists graphs for which any $(2k-1)$-spanner must have at least
$\Omega\!\left(n^{1+1/k}\right)$ edges.
Woodruff \cite{woodruff06spanners} proved that whether the conjecture is true or not, there exists a graph on $n$
nodes such that any $(2k-1)$-spanner of the graph has at least $\Omega\!\left(k^{-1}n^{1+1/k}\right)$ edges.

There are also upper and lower bounds when we allow the distortion $k$ to depend on $n$,
see
\cite{bodwinW2015verySparse,
bodwinW2016betterDistance,
chechik13spanners,
pettie2009lowDistortion}. In this paper, however, we are only interested in the case where $k = O(1)$.
The upper and lower bounds for $k = O(1)$ are summarized in Table \ref{tab:summaryExistence}.

\begin{table}[!h]
	\begin{center}
		\begin{tabular}{l|l|l|l}
			$k$ & \text{Upper Bound} & \text{Lower Bound} & \text{Reference} \\ \hline
			$2$ \& $3$ & $O\!\left(n^{3/2}\right)$ & $\Omega\!\left(n^{3/2}\right)$ &
				\cite{elkin01spanners}/\cite{wenger1991extremal} \\
			$4$ \& $5$ & $O\!\left(n^{7/5}\log^{1/5} n\right)$ & $\Omega\!\left(n^{4/3}\right)$ &
				\cite{chechik13spanners}/\cite{benson1966minimal} \\
			$\ge 6$ & $O\!\left(n^{4/3}\right)$ & $n^{4/3-o(1)}$ &
				\cite{baswana10spanners}/\cite{abboud16tightexponent} \\
		\end{tabular}
	\end{center}
	\caption{For a given $k$ an upper bound of $f(n)$ is a proof that any graph on $n$ nodes has a
	$k$-spanner with no more than $f(n)$ edges. A lower bound of $g(n)$ is a proof that there exists
	a graph on $n$ nodes for which any $k$-spanner must have at least $g(n)$ edges.}
	\label{tab:summaryExistence}
\end{table}

\subparagraph{Techniques}

Previous algorithms that construct $k$-spanners in $\Otildeh{n^2}$ time all relied on
constructing a hitting set for some set of neighbourhoods. In \cite{dor00apsp} this is done deterministically
via a dominating set algorithm, and in \cite{woodruff10nearlyquadratic} this is done via sampling. This approach
will inherently come with the cost of a $\poly \log n$ factor. Furthermore, in the construction of $6$-spanners
by Woodruff \cite{woodruff10nearlyquadratic} the number of neighbourhoods that need to be hit is so large that it seems
impossible with current techniques to modify the algorithm to be Las Vegas. Too avoid this we instead use a
clustering approach described in Section \ref{sec:clustering}. The algorithm in Theorem
\ref{thm:eightspanner} is obtained using this clustering and a careful modification of the 
path-buying algorithm of \cite{baswana10spanners}.

\subparagraph{Approximate Distance Oracles and All Pairs Almost Shortest Paths}

Given an undirected an unweighted graph $G$ an $(\alpha,\beta)$-approximate distance
oracle for $G$ is a data structure that supports the following query. Given two nodes
$u$, $v$ it can compute a distance estimate $\tilde{d}$ that satisfies
$d \le \tilde{d} \le \alpha d + \beta$ where $d$ is the distance between $u$ and $v$ in
$G$. For work on approximate distance oracles see e.g.
\cite{abrahamG2014onApproximateDistanceLabels,
agarwal2014theSpaceStretchTime,
agarwal2013aSimpleStretch,
baswanaGSU08distanceOracles,
baswanaGS09allPairs,
baswana2010faster,
baswana2006approximate,
berman2007faster,
chechik2014approximate,
chechik2015approximate,
patrascu2014distance,
patrascu2012new,
roditty2005deterministic,
sommer2014shortest,
sommer2009distance,
thorup2005approximate,
wulff2012approximate}.
Sommer \cite{sommer2016distanceOracle} gave a randomized Monte Carlo $(2,1)$-distance oracle that can be constructed in 
$\Oh{n^2 \poly \log n}$ time, has size $\Oh{n^{5/3} \poly \log n}$ and can answer queries
in $\Oh{1}$ time. We improve the construction time and the size to $\Oh{n^2}$ and $\Oh{n^{5/3}}$
respectively, and our construction is randomized Las Vegas.
As a corollary we can compute an estimate $\tilde{d}(u,v)$ for all pairs of nodes in $G$ satisfying
$d_G(u,v) \le \tilde{d}(u,v) \le 2d_G(u,v)+1$ in time $\Oh{n^2}$. This improves upon the 
$\Oh{n^2\log n}$ algorithm by Baswana and Kavitha \cite{baswana2010faster}.

%\emph{Related work.}
%On the list:
%	When $k$ is larger than a constant.
%	Fault-tolerant spanners.
%	Multiplicative spanners
%	Some for non-constant $k$.
%\nocite{*}

\subparagraph{Preliminaries}

For a graph $G$ and two nodes $u,v$ we denote the distance from $u$ to $v$ in $G$ by $d_G(u,v)$. All
graphs considered in this paper are unweighted, and unless otherwise specified they are undirected
as well. For an undirected graph $G$ an a node $u$ the neighbourhood of $u$ is the set of nodes adjacent
to $u$ and is denoted by $\Gamma_G(u)$.

\subparagraph{Overview}

In Section \ref{sec:clustering} we introduce the clustering we use when constructing
the spanners. In Section \ref{sec:construction} we show how to create an $8$-spanner with 
$\Oh{n^{4/3}}$ edges in $\Oh{n^2}$ time and thereby prove Theorem \ref{thm:eightspanner}.
In Section \ref{sec:distOracle} we provide the details on how to give an improved 
$(2,1)$-distance oracle.

\section{Clustering}
\label{sec:clustering}

Our construction of additive spanners uses standard clustering techniques.
We present our clustering framework below.
Let $G = (V,E)$ be a graph with $n$ vertices and $m$ edges. We let $t$ be a
parameter that can depend on $G$. For a sequence $u_1,\ldots,u_\ell$ of nodes
we define the clusters $C_i, i \in \set{1,\ldots,\ell}$ by
\begin{align*}
	C_i = \left ( \Gamma_G(u_i) \cup \set{u_i} \right ) \setminus \left ( C_1 \cup \ldots \cup C_{i-1} \right )
	\, .
\end{align*}
Furthermore we also define graphs $G_0,G_1,\ldots,G_\ell$ in the following way.
We let $G_0 = G$, and for $i > 0$ we let $G_i$ be the subgraph of $G$ that contains an edge
$(u,v)$ if not both $u$ and $v$ are contained in $C_1 \cup \ldots \cup C_i$.
From each node $u_i$ we let $T_i$ be a BFS tree in $G_{i-1}$ rooted at $u_i$.

\begin{definition}
	\label{def:clustering}
	A sequence $u_1,\ldots,u_\ell$ is called a \emph{$t$-clustering} if the following
	requirements are satisfied.
	\begin{itemize}
%		\item The node $u_i$ has maximal out-degree in $G_{i-1}$.
		\item The node $u_i$ maximizes
			$\left ( \Gamma_G(u_i) \cup \set{u_i} \right ) \setminus \left ( C_1 \cup \ldots \cup C_{i-1} \right )$.
		\item Every cluster $C_i$ contains at least $t$ nodes.
		\item For every node $v$ we have
				$\abs{\left ( \Gamma_G(v) \cup \set{v} \right ) \setminus \left ( C_1 \cup \ldots \cup C_{\ell} \right )} < t$.

	\end{itemize}
\end{definition}

We say that a node $v$ is \emph{clustered} if $v \in C_1 \cup \ldots \cup C_\ell$
and \emph{unclustered} otherwise. We note that since every cluster $C_i$ contains
at least $t$ nodes and the clusters are disjoint we have $\ell \le \frac{n}{t}$.

\begin{lemma}
	\label{lem:fewedgesleft}
	Let $u_1,\ldots,u_\ell$ be a $t$-clustering. Then the number of edges in $G_\ell$ is at most $nt$.
\end{lemma}
\begin{proof}
The number of edges in $G_\ell$ is bounded by the sum
$\sum_{v \in V} \abs{\left ( \Gamma_G(v) \right ) \setminus \left ( C_1 \cup \ldots \cup C_{\ell} \right )}$,
which is clearly less than $nt$.
\end{proof}

\begin{lemma}
	\label{lem:clusterOnSp}
	Let $u_1,\ldots,u_\ell$ be a $t$-clustering of $G = (V,E)$ and let $u,v \in V$
	be a pair of nodes. Assume that some shortest path from $u$ to $v$ in $G$ is not
	contained in $G_\ell$ from Lemma \ref{lem:fewedgesleft}.
	Then there exists an index $i \in \set{1,2,\ldots,\ell}$ such that 
	\begin{align*}
		d_{T_i}(u_i,u) + d_{T_i}(u_i,v) \le d_G(u,v) + 2
		\, .
	\end{align*}
\end{lemma}
\begin{proof}
Consider a shortest path $p$ from $u$ to $v$ that is not contained in $G_\ell$ and let $w$
be a clustered node on $p$ such that $w \in C_i$. We choose $w$ such that $i$ is smallest possible.
By choosing $i$ smallest possible $p$ is contained in $G_{i-1}$. Furthermore since the
distance from $w$ to $u_i$ is at most $1$ we see that
\begin{align*}
	d_{G_{i-1}}(u_i,u) + d_{G_{i-1}}(u_i,v) \le 
	d_{G_{i-1}}(w,u) + d_{G_{i-1}}(w,v) + 2 =
	d_G(u,v) + 2
	\, .
\end{align*}
Since $T_i$ is a is shortest path tree in $G_{i-1}$ the conclusion follows.
\end{proof}

\begin{lemma}
	\label{lem:fastConstruction}
	Given a graph $G$ and a parameter $t > 0$ we can construct a $t$-clustering
	$u_1,\ldots,u_\ell$, the corresponding BFS trees $T_1,\ldots,T_\ell$ and 
	$G_\ell$ in $O(n^2)$ time.
\end{lemma}
\begin{proof}
The algorithm will work by finding the nodes $u_1,\ldots,u_\ell$ consecutively, i.e.
first $u_1$, then $u_2$ and so on.
The algorithm will maintain a graph $G'$. In the beginning of the algorithm we have $G' = G_0$,
and after we add $u_i$ we will alter $G'$ such that $G' = G_i$. The total cost of altering
all $G'$ will be $O(m) = O(n^2)$.

We find $u_i$ by looking at all nodes in $G' = G_{i-1}$ and count the number of neighbours not in
$C_1 \cup \ldots \cup C_{i-1}$.
Since $G_{i-1}$ has at most $n \abs{C_i}$ edges this takes $O(n\abs{C_i})$ time. Then the algorithm
finds a BFS tree from $u_i$ in $G_{i-1}$ in $O(n\abs{C_i})$ time. Hence the total time used 
by the algorithm is:
\begin{align*}
	O \! \left (
		m + \sum_{i=1}^\ell n\abs{C_i}
	\right )
	= O(n^2)
	\, .
\end{align*}
\end{proof}

%\section{Constructing an \texorpdfstring{$O(1)$-}Spanner}
\section{Constructing $O(1)$-Spanners}
\label{sec:construction}

In this section we present our construction of an $8$-spanner with $\Oh{n^{4/3}}$ edges
in $\Oh{n^2}$ time. As a warmup we show how we can use the clustering from Section \ref{sec:clustering}
to give a $2$-spanner with $\Oh{n^{3/2}}$ edges in $\Oh{n^2}$ time.

\begin{theorem}
	\label{thm:twospanner}
	There exists an algorithm that given a graph $G$ with $n$ nodes constructs a $2$-spanner
	of $G$ with $\le 2n^{3/2}$ edges in $\Oh{n^2}$ time.
\end{theorem}
\begin{proof}
Let $t = \sqrt{n}$ and construct a $t$-clustering $u_1,\ldots,u_\ell$ with Lemma \ref{lem:fastConstruction}.
Let $H = T_1 \cup \ldots \cup T_\ell \cup G_\ell$. The number of edges in $H$ is at most
$n\ell + nt \le 2n\sqrt{n}$ by Lemma \ref{lem:fewedgesleft} and the fact that $\ell \le \frac{n}{t}$.

Now we just need to prove that $H$ is a $2$-spanner.
Let $u,v$ be arbitrary nodes and let $p$ be a shortest path from $u$ to $v$ in $G$. We wish 
to prove that
\begin{align}
	\label{eq:twospannerConstraint}
	d_H(u,v) \le d_G(u,v) + 2
	\, .
\end{align}
If $p$ is contained in $G_\ell$ then \eqref{eq:twospannerConstraint} is obviously true.
Otherwise there exists an index $i$ such that $d_{T_i}(u,v) \le d_G(u,v)+2$ by Lemma \ref{lem:clusterOnSp},
and \eqref{eq:twospannerConstraint} is true since $T_i \subset H$.
\end{proof}

Next we turn to showing how to create an $8$-spanner $H$ with $\Oh{n^{4/3}}$ edges in $\Oh{n^2}$ time.
The idea is the following. We start by creating a $t$-clustering $u_1,\ldots,u_\ell$ with
$t = n^{1/3}$ and $\ell \le n^{2/3}$. Using the BFS trees $T_1,\ldots,T_\ell$ along with
Lemma \ref{lem:clusterOnSp} we can then get an additive $2$-approximation of $d_G(u_i,u_j)$ for all
pairs of indices $i,j$, which we will call $\delta_{i,j}$. The calculation of the BFS trees in
$\Oh{n^2}$ time relies on an idea similar to one in \cite{aingworth99fastestimationshortestpath}.
The BFS trees also gives
us a path from $u_i$ to $u_j$ that is at most $2$ longer than the shortest path. If we add
all these shortest paths to our spanner along with $G_\ell$ and the neighbours in $C_i$ of
each $u_i$ we will get a $6$-spanner. Unfortunately, adding a path could require adding up
to $\Omega(\ell)$ edges, and since there are $\ell^2$ pairs we can only guarantee that the spanner has
$\Oh{\ell^3}$ edges, which is $\Oh{n^2}$ if $\ell \approx n^{2/3}$. (We only need to add edges on the
path that are not already in $G_\ell$) Instead we use an argument similar
to the path-buying argument from \cite{baswana10spanners} and the construction from
\cite{knudsen14simplespanners}. We add the path from $u_i$ to $u_j$ unless we can guarantee
that there is already an additive $2$-approximation of this path in the spanner already. We do
this by maintaining an upper bound $\Delta_{i,j}$ on the distance from $u_i$ to $u_j$ in the spanner $H$.
We then argue that if we add a path of with $k$ edges not already in the spanner,
then there are $\Omega(k)$ pairs $u_{i'}, u_{j'}$ for which the 
upper bound $\Delta_{i',j'}$ is improved. Then, this will imply that at most $\Oh{\ell^2}$ edges are
added giving an upper bound of $\Oh{n^{4/3}}$ on the number of edges in $H$.

After this informal discussion of the construction we turn to the details. The algorithm is
given a graph $G = (V,E)$ with $n$ nodes and $m$ edges, and will return a spanner $H = (V,F)$.
Initially $F = \emptyset$ and we will add edges to $H$ such that $H$ becomes a $8$-spanner
of $G$.  The algorithm starts by creating
a $t$-clustering $u_1,\ldots,u_\ell$ with $t = n^{1/3}$ using Lemma \ref{lem:fastConstruction} in
$\Oh{n^2}$ time. Since $\ell \le \frac{n}{t}$ we have
$\ell \le n^{2/3}$. Then we add edges from $u_i$ to all nodes in $C_i \setminus \set{u_i}$
to $H$ for all $i \in \set{1,2,\ldots,\ell}$. We add at most $n$ edges this way. Then we
add all edges from $G_\ell$ to $H$. This adds at most $nt = n^{4/3}$ edges to $H$.

We give each node $u \in V$ a color $c(u) \in \set{0,1,2,\ldots,\ell}$. If $u$ is unclustered
then $u$ has color $c(u) = 0$. Otherwise $c(u) = i$ where $i$ is the unique index such that $u \in C_i$.
For each pair of indices $i,j \in \set{1,2,\ldots,\ell}$ we define 
$\delta_{i,j}$ by:
\begin{align}
	\label{eq:deltaDefinition}
	\delta_{i,j} = \min_{k \in \set{1,2,\ldots,\ell}}
		\set{d_{T_k}(u_k,u_i) + d_{T_k}(u_k,u_j)}
	\, .
\end{align}
We first note that for a choice of $i,j$ we can calculate the right hand side of
\eqref{eq:deltaDefinition} in $\Oh{\ell}$ time since we are taking the minimum over
$\ell$ different values. So in $\Oh{\ell^3}$ time the algorithm calculates $\delta_{i,j}$ for 
all pairs of indices $i,j$. Since $\ell \le n^{2/3}$ this is within the $\Oh{n^2}$ time bound.
As a consequence of Lemma \ref{lem:clusterOnSp} we get that $\delta_{i,j}$ is a good approximation
of $d_G(u_i,u_j)$, more precisely:
\begin{align}
	\label{eq:deltaProperty}
	d_G(u_i,u_j) \le \delta_{i,j} \le d_G(u_i,u_j) + 2
	\, .
\end{align}
We now define $T_i'$ to be the tree obtained from $T_i$ by contracting each edge in $G_\ell$.
Since an edge is contained in $G_\ell$ iff at least one of its endpoints is unclustered we can
construct $T_i'$ from $T_i$ in $\Oh{n}$ time. The algorithm does so for all $i \in \set{1,2,\ldots,\ell}$
in $\Oh{n\ell} = \Oh{n^{5/3}}$ time. We note that the shortest path between two nodes $u,v$ in $T_i'$
contains exactly the edges on the shortest path between $u,v$ in $T_i$ excluding the edges
that are contained in $G_\ell$.

%\emph{TODO: We note that the diameter of $T_i'$ can be $\Omega(n)$, but the paths in $T_i'$ we
%will consider all have length $\le 5\ell$.}

The algorithm initializes $\Delta_{i,j} = \infty$ for all pairs of indices $i,j$ with $i \neq j$
and let $\Delta_{i,i} = 0$ for all $i$. We will maintain
that $\Delta_{i,j}$ is an upper bound on $d_H(u_i,u_j)$ throughout the algorithm.
Now the algorithm goes through all pairs $u_i,u_j$ and adds an almost-shortest path between the
nodes if needed. Specifically, we do the following:

\begin{figure}[h!]
%\begin{lstlisting}[mathescape=true,title={Algorithm 1},captionpos=t,abovecaptionskip=-\medskipamount,label=list:8-6]
\begin{lstlisting}[title={\textbf{Algorithm 1}},label=list:8-6,captionpos=t,mathescape=true,abovecaptionskip=-\medskipamount]
For each pair of indices $i,j \in \set{1,2,\ldots,\ell}$:
    For all $k \in \set{1,2,\ldots,\ell}$:
        Set $\Delta_{i,j} := \min\set{\Delta_{i,j},\Delta_{i,k}+\Delta_{k,j}}$.
    If $\Delta_{i,j} > \delta_{i,j} + 2$ do:
        Find a $k \in \set{1,2,\ldots,\ell}$ such that $d_{T_k}(u_k,u_i)+d_{T_k}(u_k,u_j) = \delta_{i,j}$.
        Find the path $p$ from $u_i$ to $u_j$ in $T_k'$.
        Add all edges from $p$ to $H$.
        Write $p = (w_0,w_1,w_2,\ldots,w_{s-1})$.
        For all $x \in \set{0,1,2,...,s-1}$:
        	Set $y := d_{T_k}(u_i,w_x)$.
            Set $\Delta_{i,c(w_x)} := \min \set{\Delta_{i,c(w_x)},y+1}$
            Set $\Delta_{c(w_x),j} := \min \set{\Delta_{c(w_x),j},(\delta_{i,j}-y)+1}$
\end{lstlisting}
	\label{alg:addPaths}
\end{figure}

Let $L$ be an upper bound on the number of nodes of the path $p$ from $u_i$ to $u_j$ in $T_k'$ on
line 6. Then Algorithm 1 can implemented in $\Oh{\ell^3 + \ell^2L}$ time.
Hence we just need to prove that $L = \Oh{\ell}$ in order to conclude that it can be implemented in
$\Oh{\ell^3} = \Oh{n^2}$ time. This follows from the fact that $p$ is an almost shortest path and
the following reasoning.
If $p$ contained $> C\ell$ nodes for some sufficiently large constant $C$ it would contain more than $C$ nodes
of the same color. Since nodes of the same color have distance at most $2$ in $G$ this would imply
that there was a much shorter path from $u$ to $v$ in $G$ contradicting \eqref{eq:deltaProperty}
if $C$ was chosen large enough.
The details with $C=5$ are given in the following lemma:
\begin{lemma}
	\label{lem:colorsOnPath}
	The path $p$ contains no nodes of color $0$, and at most $5$ nodes of each color $\neq 0$.
\end{lemma}
\begin{proof}
Obviously $p$ does not contain
a node with color $0$, since all its incident edges would be contained in $G_\ell$ and hence not
in $T_k'$. Now assume for the sake of contradiction that $p$ contains $6$ nodes of some color
$r \neq 0$. When traversing $p$ from $u_i$ to $u_j$ let $\alpha$ and $\beta$ be the first and
the last node of color $r$ respectively. The distance from $\alpha$ to $\beta$ when following
$p$ must be at least $5$ by assumption. On the other hand $\alpha$ and $\beta$ have distance
at most $2$ in $G$. So there exists a path in $G$ from $u_i$ to $u_j$ that is at least $3$ edges
shorter that $p$. This contradicts \eqref{eq:deltaProperty}. Hence the assumption was wrong
and $p$ contains at most $5$ nodes of each color $\neq 0$.
\end{proof}

Since there are $\ell$ different colors $\neq 0$ the path $p$ contains at most $5\ell$ nodes
and the running time of Algorithm 1 is $\Oh{n^2}$.
So now we just need to prove that $H$ is an $8$-spanner and that $H$ has at most $\Oh{n^{4/3}}$ edges.
We start by proving that $H$ is an $8$-spanner. Here we will utilize that the $\Delta_{i,j}$ is
an upper bound on the distance from $u_i$ to $u_j$ in $H$. Furthermore, Algorithm 1 guarantees
that $\Delta_{i,j} \le \delta_{i,j}+2$. Together with \eqref{eq:deltaProperty} this gives that
\begin{align}
	\label{eq:DeltaProperty}
	d_H(u_i,u_j) \le d_G(u_i,u_j) + 4
	\, .
\end{align}

%\emph{TODO: Add some intuition}

\begin{lemma}
	\label{lem:HIsSpanner}
	The subgraph $H$ of $G$ is an additive $8$-spanner of $G$.
\end{lemma}
\begin{proof}
Assume for the sake of contradiction that $H$ is not an additive $8$-spanner and let $u,v$
be a pair of nodes with shortest possible distance in $G$ such that:
\begin{align}
	\label{eq:contradictionHSpanner}
	d_H(u,v) > d_G(u,v) + 8
	\, .
\end{align}
Say that $d_G(u,v) = D$ and let $p = (w_0,w_1,\ldots,w_D)$ be a shortest path from $u$ to
$v$ in $G$ where $w_0 = u$ and $w_D = v$. Since the pair $(u,v)$ has the smallest possible distance
in $G$ such that \eqref{eq:contradictionHSpanner} holds and $d_G(w_1,v) = D-1$ we have
$d_H(w_1,v) \le (D-1)+8$. In particular the edge $(u,w_1)$ is not in $H$ as it would
contradict \eqref{eq:contradictionHSpanner}. Hence $u$ cannot be unclustered, as all the
edges incident to an unclustered node is contained in $G_\ell$ and therefore $H$. With the same
reasoning we conclude that $v$ is clustered. Let the colors of $u$ and $v$ be $i$ and $j$
respectively. The distances from $u$ and $v$ to $u_i$ and $u_j$ respectively are at most $1$.
Combining this insight with \eqref{eq:DeltaProperty} we get:
\begin{align}
	\notag
	d_H(u,v) \le 
	d_H(u_i,u_j) + 2 \le 
	d_G(u_i,u_j) + 6 \le 
	d_G(u,v) + 8
	\, .
\end{align}
But this contradicts the assumption \eqref{eq:contradictionHSpanner}. Hence the assumption was wrong
and $H$ is an additive $8$-spanner of $G$.
\end{proof}

Lastly, we need to prove that $H$ contains no more than $\Oh{n^{4/3}}$ edges. Informally, we argue
the following way. Whenever the $s-1$ edges of $p$ are added to $H$ on line 7 of Algorithm 1 there
are $\Omega(s)$ different colors on $p$. For each color $r$ on $p$ we then argue that either
$\Delta_{i,r}$ or $\Delta_{r,j}$ are made smaller on line 11 or 12 of Algorithm 1. Lastly, we
argue that $\Delta_{i,j}$ can only be updated $\Oh{1}$ times, and since there are $\ell^2 \le n^{4/3}$
variables $\Delta_{i,j}$ this implies that Algorithm 1 only adds $\Oh{n^{4/3}}$ edges to $H$.
This intuition is formalized in Lemma \ref{lem:algAddsFewEdges} bellow:
\begin{lemma}
	\label{lem:algAddsFewEdges}
	Algorithm 1 adds no more than $25\ell^2$ edges to $H$.
\end{lemma}
\begin{proof}
Say that the algorithm adds the edges from the path $p = (w_0,w_1,\ldots,w_{s-1})$ on line 7
of Algorithm 1 where $w_0=u_i, w_{s-1}=u_j$. First we note that since $d_G(u_i,u_j) \ge \delta_{i,j}-2$
by \eqref{eq:deltaProperty} we have that $d_G(u_i,w_x) \ge y-2$ for every
$x \in \set{0,1,\ldots,s-1}$, where we consider $y$ to be a function of $x$ defined by
$y = d_{T_k}(u_i,w_x)$ as on line 10. Now fix $x$ and let
$r = c(w_x)$. Then there is an edge between $w_x$ and $u_r$ and therefore
$d_G(u_i, u_r) \ge y-3$, i.e. $y+1 \le d_G(u_i,u_r)+4$. So if Algorithm 1 decreases $\Delta_{i,r}$
on line 11 we have $\Delta_{i,r} \le d_G(u_i,u_r)+4$ after it is decreased. Since $\Delta_{i,r}$
is an upper bound on $d_H(u_i,u_r)$ and therefore also an upper bound on $d_G(u_i,u_r)$ we see that
$\Delta_{i,r}$ can be decreased at most $5$ times for each choice of $i,r$. By symmetry we see that
we can also decrease $\Delta_{r,j}$ on line 12 at most $5$ times. Since there are $\ell^2$ pairs of
indices the algorithm can change the values of $\Delta_{i,r}$ or $\Delta_{r,j}$ on line 11 and 12
of Algorithm 1 at most $5\ell^2$ times.

Let $r$ be a color on $p$. After the execution of lines 9-12 we have
\begin{align}
	\notag
	\Delta_{i,r} + \Delta_{r,j} \le
	\delta_{i,j} + 2
	\, .
\end{align}
Due to the execution of lines 2 and 3 this was not the case before. Hence either $\Delta_{i,r}$
or $\Delta_{r,j}$ were updated. By Lemma \ref{lem:colorsOnPath} there are at least $\frac{s}{5}$ colors
on $p$, so if the algorithm adds $A$ edges in total it makes at least $\frac{A}{5}$ updates of
upper bounds $\Delta_{i,r}$ or $\Delta_{r,j}$. Since there can be at most $5\ell^2$ such updates
we conclude that $\frac{A}{5} \le 5\ell^2$ and that Algorithm 1 adds no more than $5\ell^2$
edges.
\end{proof}

To summarize, the algorithm presented in this section runs in $\Oh{n^2}$ time and gives an additive
$8$-spanner with no more than $26n^{4/3}+n = \Oh{n^{4/3}}$ edges. We have made no attempt to optimize
the constant in the $O$-notation. Hence we get:

\begin{theorem}
	\label{thm:eightspanner}
	There exists an algorithm that given a graph $G$ with $n$ nodes constructs a $8$-spanner
	of $G$ with $\Oh{n^{4/3}}$ edges in $\Oh{n^2}$ time.
\end{theorem}
\section{Distance Oracles}
\label{sec:distOracle}

In the following we show how to modify the construction by Sommer \cite{sommer2016distanceOracle}
to obtain a $(2,1)$-distance oracle of size $\Oh{n^{5/3}}$ that can be constructed in expected
$\Oh{n^2}$ time.

Let $G$ be a given graph, and $H$ an $8$-spanner of $G$ constructed by Theorem \ref{thm:eightspanner}.
$H$ is constructed in $\Oh{n^2}$ time and has $\Oh{n^{4/3}}$ edges. During the construction we use
only $\Oh{n^{5/3}}$ space.

Let $u_1,u_2,\ldots,u_\ell$ be a $n^{1/3}$-clustering of $G$. 
Using Lemma \ref{lem:fastConstruction} we obtain $T_1,\ldots,T_\ell$ and $G_\ell$ in $\Oh{n^2}$ time.
For each node $v$ we define four \emph{portals} $p_1(v), p_2(v), p_3(v), p_4(v)$. We define
$p_1(v) = u_i$, where $u_i$ is chosen such that the distance between $v$ and $u_i$ in $T_i$ is minimized.
In case of ties we choose the node $u_i$ with the lowest index $i$. The node $p_{j+1}(v)$ for
$j=1,2,3$ is chosen depending on $p_j(v)$. If $p_j(v) = u_1$ we let $p_{j+1}(v) = u_1$. Otherwise
$p_j(v) = u_i$ for some index $i$. We let $p_{j+1}(v) = u_{i'}$ where $u_{i'}$ is chosen among
$u_1,u_2,\ldots,u_{i-1}$ such that the distance between $u_{i'}$ and $v$ in $T_{i'}$ is minimized.
In case of ties we choose the node $u_{i'}$ with the lowest index $i'$.
The portals for all nodes can be found in $\Oh{n^{5/3}}$ time.

We will use the following lemma by  P\v{a}tra\c{s}cu and Roditty \cite{patrascu2014distance}
to construct a $(2,1)$-distance oracle for $G_\ell$, that uses space $\Oh{n^{5/3}}$.
\begin{lemma}[\cite{patrascu2014distance}]
	\label{lem:patrascuRoditty}
	For any unweighted, undirected graph, there exists a distance oracle of size $\Oh{n^{5/3}}$ that,
	given any nodes $u$ and $v$ at distance $d$, returns a distance of at most $2d+1$ in constant time.
	The distance oracle can be constructed in expected time $\Oh{mn^{2/3}}$.
\end{lemma}
In the proof in \cite{patrascu2014distance} they only claim a running time of $\Oh{mn^{2/3} + n^{7/3}}$,
however, this can be fixed to give the correct running time of $\Oh{mn^{2/3}}$ \cite{rodittyPersonalCom}.
By \cite[Claim 9]{patrascu2014distance} it is easy to see how to get a running time of $\Oh{mn^{2/3}+n^2}$
which suffice for our purposes.

We are now ready to define the distance oracle. For each $i = 1,2,\ldots,\ell$ we store the distances
$d_{T_i}(u_i, v)$ and $d_{H}(u_i, v)$ for all nodes $v$. The distances $d_H(u_i,v)$ can be calculated
using a BFS in time $\Oh{\ell n^{4/3}} = \Oh{n^2}$. For each node $v$ we store its portals $p_j(v), j=1,2,3,4$.
We augment this distance oracle with the P\v{a}tra\c{s}cu-Roditty distance oracle from Lemma 
\ref{lem:patrascuRoditty} for $G_\ell$.

We now show how to use the distance oracle to obtain approximate distances for a query $u,v$. We let
$\delta_{PR}(u,v)$ be the approximate distance in $G_\ell$ returned by the P\v{a}tra\c{s}cu-Roditty 
distance oracle. We define $\delta_j(u,v)$ in the following way. Let $p_j(u) = u_i$. Then 
$\delta_j(u,v) = d_{T_i}(u_i,u) + \min \set {d_{T_i}(u_i,v), d_{H}(u_i,v)}$. The distance returned
by the distance oracle is the minimum of $\delta_{PR}(u,v)$, $\delta_j(u,v)$ and $\delta_j(v,u)$
for $j=1,2,3,4$.

We will now argue that if the the distance between $u$ and $v$ is $d$, then the distance oracle
returns a distance between $d$ and $2d+1$. The distance returned is obviously at least $d$, so
we just need to show that it is at most $2d+1$. Consider a shortest path between $u$ and $v$ in $G$.
If there is at most one
node on the shortest path which is incident to a node $u_i$ in the clustering then the shortest
path is contained in $G_\ell$, and therefore:
\begin{align*}
	\delta_{PR}(u,v) \le 2d_{G_\ell}(u,v) + 1
	= 2d+1
	\, .
\end{align*}
So assume that there exists a edge on the shortest path not in $G_\ell$. Let $i$ be the smallest 
index such that there is an edge $(z,t)$ on the shortest path with $z,t \in C_1 \cup \ldots \cup C_i$. Say that
$z$ is closer to $u$ than to $v$ in $G$. Assume that $z \in C_i$ and $t \in C_{i'}$ for some index $i' \le i$
(the case where $z \in C_{i'}$ and $t \in C_{i}$ is handled symmetrically).
Since the shortest path is contained in $G_{i-1}$ and $G_{i'-1}$ we have that
$d_{T_i}(u_i,u) + d_{T_{i'}}(u_{i'},v) \le d+1$ and therefore:
\begin{align*}
	\min \set{ d_{T_i}(u_i,u), d_{T_i}(u_i,v)} \le \frac{d+1}{2}
	\, .
\end{align*}
Assume that $d_{T_i}(u_i,u) \le \frac{d+1}{2}$. The other case is handled similarly.
Say that $p_j(u) = u_{k_j}$ for $j=1,2,3,4$. 
First assume that $k_j > i$ for all $j=1,2,3,4$. Then we conclude that 
$d_{T_{k_1}}(p_1(u),u) \le d_{T_i}(u_i,u) - 4$. The distance returned by the distance oracle is
at most
\begin{align*}
	\delta_1(u,v)
	& \le d_{T_{k_1}}(p_1(u),u) + d_H(p_1(u),v)
	\\ & \le 
	2d_{T_{k_1}}(p_1(u),u) + d_H(u,v)
	\\ & \le
	2(d_{T_i}(u_i,u)-4) + d + 8
	\le 
	2d + 1
	\, .
\end{align*}
Now assume that $k_j \le i$ for some $j \in \set{1,2,3,4}$ and let $j$ be the smallest index such that $k_{j} \le i$.
By definition we have that $d_{T_{k_{j}}}(p_j(u),u) \le d_{T_i}(u_i,u)$. Furthermore the shortest path
is contained in $G_{j-1}$ and therefore $d_{T_{k_{j}}}(p_j(u),v) \le d_{T_{k_{j}}}(p_j(u),u) + d_G(u,v)$.
The distance returned is at most
\begin{align*}
	\delta_j(u,v)
	& \le 
	d_{T_{k_{j}}}(p_j(u),u) + d_{T_{k_{j}}}(p_j(u),v)
	\\ & \le 
	2d_{T_{k_{j}}}(p_j(u),u) + d
	\\ & \le 
	2d_{T_i}(u_i,u) + d
	\le
	2d+1
	\, .
\end{align*}
We conclude that the distance returned by the distance oracle is always between $d$ and $2d+1$. The
result is summarized in Theorem \ref{thm:distanceOracle}.
\begin{theorem}
	\label{thm:distanceOracle}
	For any unweighted, undirected graph, there exists a distance oracle of size $\Oh{n^{5/3}}$ that,
	given any nodes $u$ and $v$ at distance $d$, returns a distance of at most $2d+1$ in constant time.
	The distance oracle can be constructed in expected time $\Oh{n^2}$.
\end{theorem}

\subparagraph*{Acknowledgements.}

The author would like to thank Christian Sommer for helpful discussions on the application
of the $8$-spanner to his construction of distance oracles.

\newpage

\bibliographystyle{plain}
\bibliography{spanner}

%\makeatletter
%\def\runninghead{\hrulefill\quad APPENDIX\quad\hrulefill}
%\def\ps@headings{
%	\def\@oddhead{\footnotesize\rm\hfill\runninghead\hfill}}
%\def\@evenhead{\@oddhead}
%\def\@oddfoot{\rm\hfill\thepage\hfill}\def\@evenfoot{\@oddfoot}
%\makeatother
%
%\newpage
%\setlength{\headsep}{15pt} \pagestyle{headings}
%
%\appendix

%Appendix?

\end{document}